\documentclass[letterpaper, 10 pt, conference]{ieeeconf}
\IEEEoverridecommandlockouts

\usepackage{cite}
\usepackage{graphicx}
\usepackage{amsmath}
\usepackage{algorithm}
\usepackage{algpseudocode}

\usepackage{amssymb}
\usepackage{color}
\usepackage[mathscr]{euscript}
\usepackage{fixmath}
\usepackage{url}
\usepackage{tkz-graph}
\usepackage{mdframed}
\usepackage{siunitx}
\usepackage{relsize}
\usepackage{subcaption}
\usepackage{bm}
\usepackage{siunitx}
\usepackage{wrapfig}
\usepackage{dsfont}
\usepackage{enumerate}
\usepackage{balance}
\usepackage{empheq}
\usepackage{mathtools}
\IEEEoverridecommandlockouts     
\allowdisplaybreaks
\newcommand{\R}{\mathbb{R}}
\newcommand{\inlineeqnum}{\refstepcounter{equation}~~\mbox{(\theequation)}}

\newtheorem{theorem}{Theorem}

\newtheorem{definition}{Definition} 

\title{\LARGE \bf
	Privacy of Real-Time Pricing in Smart Grid
}

\author{Mahrokh GhoddousiBoroujeni$^{1}$, Dominik Fay$^{2}$, Christos Dimitrakakis$^{3}$ and Maryam Kamgarpour$^{4}$% <-this % stops a space
\thanks{This research was gratefully funded by the European Union ERC Starting Grant CONENE.}%
\thanks{$^{1}$is with Department of Electrical Engineering, Sharif University of Technology, Iran. e-mail: {\tt\footnotesize ghoddousi\_mahrokh@ee.sharif.edu}}%
\thanks{$^{2}$is with the EECS school, KTH Royal Institute of Technology, Sweden. e-mail: {\tt\footnotesize dominikf@kth.se}}%
\thanks{$^{3}$is with the Computer Science and Engineering faculty, Chalmers University of Technology, Sweden. e-mail: {\tt\footnotesize chrdimi@chalmers.se}}%
\thanks{$^{4}$is with the Automatic Control Laboratory, D-ITET, ETH Zurich, Switzerland. e-mail: {\tt\footnotesize mkamgar@control.ee.ethz.ch}}%
\thanks{This research was conducted while the first and second authors were doing an internship in D-ITET and Chalmers University of Technology.}%
}	

\begin{document}
\maketitle	
\IEEEpeerreviewmaketitle
%%%%%%%%%%%%%%%%%%%%%%%%%%%%%%%%%%%%%%%%%%%%%%%%%%%%%%%%%%%%%%%%%%%%%%%%%%%%%%%%
\begin{abstract}
Installing smart meters to publish real-time electricity rates has been controversial while it might lead to privacy concerns. 
Dispatched rates include fine-grained data on aggregate electricity consumption in a zone and could potentially be used to infer a household's pattern of energy use or its occupancy. 
In this paper, we propose Blowfish privacy to protect the occupancy state of the houses connected to a smart grid.
First, we introduce a Markov model of the relationship between electricity rate and electricity consumption.
Next, we develop an algorithm that perturbs electricity rates before publishing them to ensure users' privacy.
Last, the proposed algorithm is tested on data inspired by household occupancy models and its performance is compared to an alternative solution.  
\end{abstract}

%%%%%%%%%%%%%%%%%%%%%%%%%%%%%%%%%%%%%%%%%%%%%%%%%%%%%%%%%%%%%%%%%%%%%%%%%%%%%%%%
\section{INTRODUCTION}
% Motivation
Smart energy meters allow us to establish a two-way interactive system between consumers and energy providers through fetching users' power consumption in discrete time intervals and assigning an electricity rate based on demand and generation costs. This concept is known as \textit{Real-Time Pricing (RTP)}. 
The fluctuating rate enables users to adjust their consumption accordingly and shift it from high demand periods, at a higher rate, to low demand periods and save costs. From the provider's perspective, this scheme can help flatten the demand peaks and utilize the generation of power, specially renewables, in a more effective way. 

%{\color{cyan} (SM privacy concern)}
Even if consumers trust the electricity provider, publicly broadcast rates leak information about individuals, such as their occupancy and consumption patterns. %{\color{cyan} (goal)}
We provide an algorithm for determining the rates retaining the utility of real-time pricing method while preserving the privacy of individual houses' occupancy. 
In this study, it is assumed that the curator, who broadcasts the electricity rates, is trusted and the problem of how to publish privacy-preserving rates that maintain the utility of the RTP scheme is considered.  A summary of the problem under study is depicted in Fig.~\ref{fig:setting}.

%{\color{cyan}(DP intro)}
\paragraph{Related work}
We use a generalization of the mathematically rigorous notion of \textit{differential privacy (DP) c.f.~\cite{DPbook}}.
DP allows publishing statistical queries on a dataset so that no third party can infer much about any respondents. A fundamental approach in achieving DP is randomized responses to queries. However, while more noise better hides data, they reduce the accuracy and hence, the usefulness of the published query answers. This problem is referred to as the privacy-utility trade-off. 

%{\color{cyan}(correlated data)}
In typical applications of DP, the dataset is in a matrix form, with each row corresponding to a user. Many real applications, such as smart metering, however, involve time series data. Then a user may be associated with multiple rows, hence increasing the possibility of privacy violation when DP is applied naively\cite{rig}. 
A straightforward solution is using group DP~\cite{DPbook} but this may add too much noise to be useful.
%Past work has tried to enhance the utility of this approach by adding post-processing modules\cite{adap}.  
An alternative solution is to transform, e.g. using Discrete Fourier Transform\cite{DFT} or Wavelet Transform\cite{deploy}, a time-series dataset, so as to reduce each individual's data to a single row. However, these methods require the whole time-series in advance and are not useful for real-time applications such as smart metering.

%{\color{cyan}(lit review: SM privacy)}
Smart meters are an interesting application scenario for privacy. One previous approach considered installing batteries that charge or discharge randomly for each house~\cite{bat} but did not consider the time series element and required extra facilities, i.e. batteries.
% pufferfish
A recent study~\cite{deploy} used Wavelet transform to avoid third parties with access to the fine-grained smart meter's data from obtaining information about the state of use of devices at each house. 
%{\color{cyan}(What is missing)}
   %{\color{cyan}(what do we model?)}
The state of the art approach in dealing with complex data acquisition scenarios is to capture the relationships between data and users with a model.
One idea is to define a graphical model that ties observations to secrets that must be kept~\cite{cor}. Although this approach is very promising, it scales exponentially. We instead are using the concept of Blowfish privacy.
%Some other work has looked at the case of hidden Markov models specifically\cite{dphmm} and applied this to a road location dataset\cite{loc}. Unfortunately the modelling assumptions made in this work are not applicable to our setting.

% Goals

%Contributions
\paragraph{Our contributions.}
First, we introduce a model from the Markov family describing the correlation of households' occupancy states over time.
Our model, in contrast to past work, has outputs in a continuous domain and a probabilistic mapping between the states and outputs.
Second, we design a mechanism to provide individual privacy against adversaries monitoring publicly available rates. This is unlike previous works, which only consider the privacy of smart meter data, which we can assume to be confidential between the household and the data provider anyway. Our approach, in contrast to past work, considers the temporal correlations in the dataset and outputs a result in each time step, consistent with the setting of a private RTP scheme. Third, we evaluate our method using synthetic electricity consumption data.

%                                        ****************Notation****************
\textit{Notation:} We denote scalar variables, vectors, and matrices by x, X and \textbf{X} respectively. 
$X_i$ is the \textit{i}'th element in vector X and $\mathbf{X}_{ij}$ is the element in row i and column j of \textbf{X}.
Given a vector $X$ we refer to its norm $1$ by $||X||_1$.
Subscript \textit{i} and superscript \textit{t} refer to every variable related to household \textit{i} at time step \textit{t}. 

%*******************************************************************************
%*******************************PROBLEM STATEMENT*******************************
%*******************************************************************************
\section{PROBLEM STATEMENT}
%                 **************************************Problem Statement******************************
\subsection{Real-Time Pricing Setting}\label{RTP_Setting}
%{\color{cyan}(User end)}
The real-time pricing (RTP) scheme has been proposed as a method to adjust the power balance between supply and demand in smart grid systems\cite{RTP}. 
A trusted data curator, who may be electricity provider, publishes electricity rates in $T\in\mathbb{N}$ successive time steps based on total demand, including the demand of $N\in\mathbb{N}$ houses connected to a smart grid.
We denote the released rate at time step $t\in\mathbb{N}$ by $\hat{r}^t\in\R_+$. 

Each household \textit{i} uses an $X^t_i \in \R_+$ amount of energy in time step \textit{t} depending on a number of previously released rates, $\hat{r}^{t'}, t'{<}t$, and an occupancy indicator $S_i^t \in \{0,1\}$, where $S_i^t = 1$ denotes that household $i$ is occupied at $t$. 
The consumption amount is unknown to the curator or a third party, however, we assume it is described by a conditional probability density $f(X^t_i|S^t_i,\hat{r}^{t'})$.
A time-variant density function captures the geographical effects on electricity consumption and can be used in the rest of this paper, but we suppress this explicit dependency for notation simplicity.
We also assume different houses have independent consumption values and $X_i^t$ is upper-bounded by $u_i \in \R_+$.

%{\color{cyan}(Pricing)}
The aggregate electricity consumption, $Z^t = \sum_{i=1}^N X^t_i \in \R_+$, is sent back to the trusted data curator. 
In this work, we assume the total demand is equal to this aggregate consumption and responds to the price signal, that is, it is elastic. 
Adding an inelastic term can readily be accounted for in our setting.
Consider the cost of generation as a strongly convex quadratic function $J(Z^t) = \frac{\alpha}{2} {(Z^t)}^2 {+} \beta Z^t {+} \gamma$, a commonly used model in energy optimization work\cite{cost}. 
Given the total demand $Z^t$, the generator aims to minimize the objective function $J(Z^t)-r^tZ^t$ by setting an electricity rate $r^t\in\R_+$. 
Using the first-order optimality condition we obtain that the optimal rate must satisfy $r^t=\alpha Z^t + \beta \quad \alpha,\beta \in \R_+ \inlineeqnum \label{eq:fr}$.

A summary of the problem setup is shown in Fig.~\ref{fig:setting}.
A third party observing the optimal electricity rates can obtain information about the total consumption, as a result of \eqref{eq:fr}.
Consequently, she can exploit the aggregated data to make inferences about the individual household's occupancy~\cite{DPbook}. 
While this is an absolute violation of the individuals' privacy, the electricity rate is public information and cannot be protected by data security methods that translate the data to other forms, such as encryption.  
Alternatively, the trusted curator employs a \textit{privacy-preserving mechanism} to perturb $r^t$ and publishes its result, $\hat{r}^t$.

While private mechanisms are well-studied for such aggregation of data, there are two major difficulties in the content of RTP in power grids: 
First, the process is repeated over time and the households' occupancy at each time are correlated.
Second, what needs to be kept secret is the occupancy, which is different from the individual consumption rates sent to the curator.
\begin{figure*}[t!]
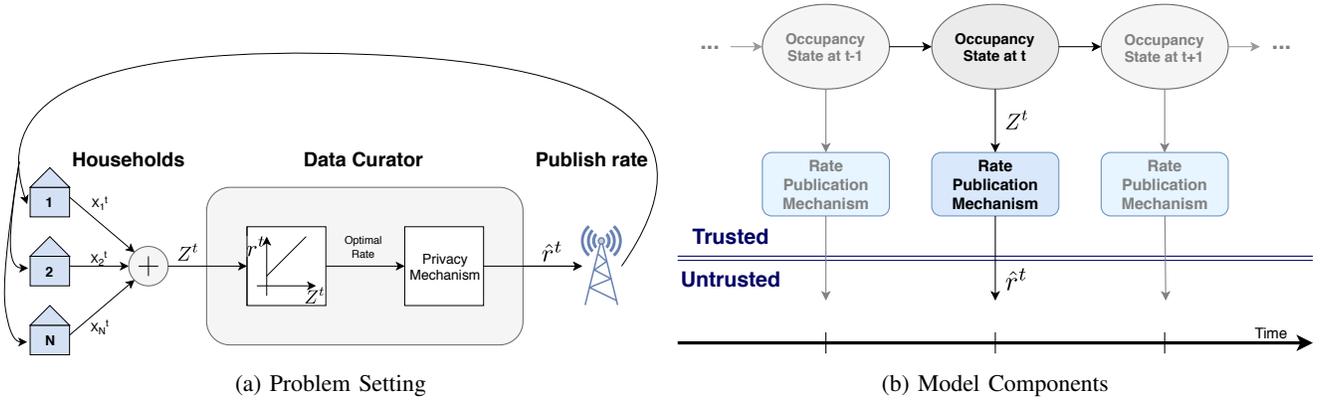

	\centering
	\begin{subfigure}[t]{0.49\linewidth}
	    \centering
		\includegraphics[width=\linewidth]{Setting.png}
		\caption{Problem Setting}
		\label{fig:setting}
	\end{subfigure}
	\begin{subfigure}[t]{0.49\linewidth}
	    \centering
		\includegraphics[width=\linewidth]{Model.png}
		\caption{Model Components}
		\label{fig:model}
	\end{subfigure}

	\caption{Problem Statement}
\end{figure*}

%***********************************************Data Generating Model******************************************************
\subsection{Data Generation Model} \label{model}
Broadcasting an electricity rate close to the optimal one has the utility of optimizing the electricity provider's profit.
One can add sufficient perturbation or stochastic noise to the output (in this case the rates) to ensure privacy, but clearly, this reduces the utility.
%{\color{cyan}(Motivation)}
As shown by\cite{lunch, onthe}, it is not possible to provide sufficient privacy and utility without making any assumptions about how the data is generated.
%{\color{cyan}(What to model)}
Hence, we need to take into account how people change their occupancy state, $S^t=(S^t_1,\cdots,S^t_N)\in\{0,1\}^N$, and how the occupancy state influences the consumption.
%{\color{cyan}(Markov)}
Based on past work on occupancy modeling\cite{learnHMM,occHMM,wilhelm}, we suppose $S^t$ forms a Hidden Markov Model (HMM) over time with the released rates as outputs and the following details, as in Fig.~\ref{fig:model}.
\begin{itemize}
	\item \textbf{States ($S^t\in\{0,1\}^N$):}
	Each state represents the occupancy of all houses under consideration. 
	There are $2^N$ different states denoted by $(M^1,\cdots,M^m)$ where $m=2^N$ and $M^i\in\{0,1\}^N~\forall i \in \{1,\cdots,m\}$.
	
	\item \textbf{Initial state probability distribution ($\Pi \in [ 0,1 ]^m$):} $\Pi = (\pi_1,\cdots,\pi_m)$ with $\pi_i = Pr[S^0 = M^i]$ is the probability distribution over all occupancy states when no rate is published (where $S^0$ is the initial state and \textit{Pr[$\cdot$]} denotes the probability of its argument).
	
	\item \textbf{Transition matrix ($\mathbf{A}^t\in \R_+^{m{\times} m}$):} How people change their occupancy state is modeled by $\mathbf{A}^t$, where $
	\mathbf{A}^t_{ij} = Pr[S^{t+1}=M^j |S^t=M^i]$.
	\item \textbf{Observation probability distribution:} 
	In contrary to the previously suggested models~\cite{wilhelm,learnHMM,occHMM} with $Z^t$ as the output, we embed a private mechanism in the HMM to take $Z^t$ as input and return $\hat{r}^t$. In this case, the observations, or model outputs, are the released rates and the observation probability distribution can be calculated from $f(X^t_i|S^t_i,\hat{r}^{t'})$ and the definition of $Z^t$.  
\end{itemize}
As the outputs of the model above are in the continuous domain $\R_+$, this model falls into the framework of a Continuous Density Hidden Markov Model (CDHMM).

%{\color{cyan}(Set of models)}
Given limited prior knowledge of the initial state distribution or the transition matrix, we consider a finite set of CDHMMs, $\Theta=\{\theta{=}(\Pi^\theta,\mathbf{A}^{t, \theta})\}$, as suggested by\cite{cor}.
All models in $\Theta$ share an identical set of states and observation probability density function, but have an exclusive transition matrix $\mathbf{A}^{t, \theta}$ and initial state distribution $\Pi^\theta$. 

%                 **************************************Adversarial Model******************************
\subsection{Adversarial Model} \label{adv}
% Trusted curator
Every third party who wishes to increase her chance of guessing a single participant’s occupancy state is called an \textit{adversary}.
We consider the case in which the adversary does not have access to the consumption values, as in Fig.~\ref{fig:model}. 
We also assume that the number of houses in the grid is fixed and known to the adversary and she considers a set of possible models described by $\Theta$. 

Let $P^{t|t{-}1,\theta}\in\R^{m}$ represent a prior probability distribution over the states of the model $\theta \in \Theta$ before observing the output at time \textit{t}. 
Note that $\Pi^\theta$ is the same as the prior at $t{=}1$, i.e. $\Pi^\theta {=} P^{1|0, \theta}$.
% posterior
After observing the model's output at $t$, we update the states' distribution by Bayesian inference to obtain the posterior probabilities $P^{t|t, \theta}\in\R^{m}$. 
The prior and posterior probabilities are related in consecutive time steps by,
\begin{equation}
	\label{eq:pri}
	P^{t|t{-}1,\theta} = P^{t{-}1|t{-}1, \theta} \times \mathbf{A}^{t, \theta},
\end{equation}	
\begin{equation}
	\label{eq:post}
	P^{t|t, \theta}_i = \frac{P^{t|t{-}1, \theta}_i \cdot Pr[\hat {r}^t|M^i]}{\sum_{j=1}^m {P^{t|t{-}1, \theta}_j \cdot Pr[\hat{r}^t|M^j]}} ~\forall i \in \{1,\cdots,m\}.
\end{equation}

We assume the model class $\Theta$ has been learned in advance from electricity consumption data. 
See, \cite{learnHMM} on applying machine learning to learn the occupancy model based on consumption data.
An adversary, who knows $\Theta$, may make strong conclusions, such as verifying a house is not occupied, by tracking the published rates. 
In the next section, we design a mechanism to ensure privacy considering such potential knowledge acquisition due to the continuous rate broadcast.
%*******************************************************************************
%*******************************************************************************
%*******************************************************************************
\section{ANALYSIS} \label{analysis}
In this section, first, we define differential privacy (DP) mathematically.
Next, we introduce a generalization of DP, Blowfish privacy (BP), and develop our instantiation of it.
Finally, we propose a private electricity rating mechanism based on the designated privacy notion and the data generation and adversarial models of Sections~\ref{model}, \ref{adv}.

%                              ****************Blowfish Privacy****************
\subsection{Blowfish Privacy} \label{Blowfish}
%{\color{cyan}(Privacy definition elements in our problem)}
% DP
%{\color{cyan}(Private mechanism)}
DP is the current standard on a mechanism's privacy evaluation.
A mechanism takes a database, a universe $U$ of all possible datasets, and a query as input and produces outputs (perturbed query result) by adding noise to the queries.
DP promises that every possible output is \textit{essentially} equally likely to occur, independent of the presence of a single individual\cite{DPbook}. 

\begin{definition}\label{def:DP}
	(Differential Privacy)
	Given a privacy parameter $\epsilon \in \R_+$, a mechanism $\mathcal{A} : U\rightarrow\R$ is $\epsilon$\textit{-differentially private} if for all database pairs $(D^t$, $\hat{D}^t)$ that differ in only one entry, $||D^t{-}\hat{D}^t||_1{=}1$, and every set of outputs $R \subseteq Range(\mathcal{A})$, 
	\begin{equation}
	Pr[\mathcal{A}(D^t) \in R] \leq \exp(\epsilon) Pr[\mathcal{A}(\hat{D}^t) \in R]. \label{eq:DP}
	\end{equation}
\end{definition}
The datasets in the pair $(D^t$, $\hat{D}^t)$ in~\ref{def:DP} are called \textit{neighbors}. 
%                              ****************Blowfish Privacy****************
% Why DP not enough
In this work, the query of the electricity rate is answered $T$ times on the temporally correlated occupancy datasets $S^1,\cdots,S^T$.
If correlations are not modeled, rather all entries are assumed correlated, the required noise to ensure a privacy level would be too much and the RTP utility will be reduced.
There is no possibility to contemplate any information on the correlation of datasets, such as their evolution model, using DP.
On the other hand, the BP framework~\cite{Blowfish} is capable of exploiting the data generation model described in Section~\ref{model}. 

%{\color{cyan}(Blowfish intro)}
Blowfish privacy~\cite{Blowfish} extends DP by contemplating the information that must be kept hidden from the adversary and constraints that may be known about the data.
%{\color{cyan}(Elements)} 
Information that we would like to protect about each user is named \textit{secret}, denoted by $\mathcal{S}^t$ at time $t$. 
\textit{Discriminative pairs}, $\mathcal{S}^t_{pairs}$, are properties that must be indistinguishable to the adversary. 
As we wish to hide every evidence on the occupancy of the houses, the dataset at time $t$ is the state of occupancy of every household $D^t \mathrel{\mathop:}= S^t$, the secrets are all entries of the dataset $\mathcal{S}^t \mathrel{\mathop:}= \{S^t_i | i = 1, \cdots, N\}$ and the discriminative pairs are \textit{(house i is empty at time \textit{t}, house i is full at time \textit{t})} for all households, $\mathcal{S}^t_{pairs}  \mathrel{\mathop:}= \{(S^t_i{=}0,S^t_i{=}1)| i = 1, \cdots, N\} \in \mathcal{S}^t \times \mathcal{S}^t$.

% Constraint
In this framework, adversarial knowledge is specified by a set of \textit{constraints}, $Q^t$, restricting the universe $U$  to the compatible ones.
The constraints indicate a subset of potential datasets from $U$ based on the acquired information from the correlations.
Hence, determining the possible datasets assessing the inferences made from existing correlations is the most crucial step to define the constraints.
An adversary with assumptions in Section~\ref{adv} may conclude that some states are impossible to happen observing the rates, i.e. there may exist states with zero prior probability in \eqref{eq:pri}.
Such information acquisition can be formulated by setting a constraint of $Q^t\mathrel{\mathop:}=$\textit{having non-zero prior probability at time $t$}.
%In each time interval, it may be concluded that some states are impossible to happen based on the observations made and the data generation model, i.e. there may exist states with zero prior probability in \eqref{eq:pri}.
%An adversary with assumptions in Section~\ref{adv} is sure that these states are not plausible as well, thus protecting them is pointless. 
We may need less perturbation and gain better utility from the RTP scheme by setting the constraint $Q^t$ and leaving incompatible datasets out.
\begin{definition}
	(Blowfish Neighbors) \label{neighbor}
	Given a set of discriminative pairs $\mathcal{S}^t_{pairs}$ and a set of constraints $Q^t$, $(D^t,\hat{D}^t)$ is a pair of \textit{Blowfish neighbors} if both datasets satisfy $Q^t$ and differ in only one entry belonging to $\mathcal{S}^t_{pairs}$.
\end{definition}

Defining BP relies on $\mathcal{S}^t_{pairs}$ and $Q^t$, in addition to $\epsilon$, and is different from DP in defining the neighbors.
\begin{definition} \label{privacy}
	(Blowfish Privacy) 
	Given $\epsilon \in \R_+$, discriminative pairs $\mathcal{S}^t_{pairs}$, and a set of constraints $Q^t$, a mechanism $\mathcal{A}:U \rightarrow \R$ is ($\epsilon,\mathcal{S}^t_{pairs},Q^t$)\textit{-Blowfish private} if for all Blowfish neighbor pairs $(D^t,\hat{D}^t)$ and every set of outputs $R\subseteq Range(\mathcal{A})$, Equation~\eqref{eq:DP} holds.
\end{definition}

There are two main differences between the privacy notion and the mechanism we use and the original BP framework in Definition~\ref{privacy}.
First, the dataset is a hidden variable, and not exactly the input to the mechanism, as opposed to the private mechanisms defined above. 
This stems from the fact that the query answer cannot be determined explicitly given $S^t$ as the dataset, but the total consumption value $Z^t$ is required as well.
In particular, in this case, the input to a mechanism $\mathcal{A}:[0,\sum_i^N u_i]\rightarrow\R$ is $Z^t$ but what we want to keep private is $S^t$.
Hence, if two groups of $N$ households with states of occupancy described by a Blowfish neighbor pair $(S^t,\hat{S}^t)$ consume $Z^t$ and $\hat{Z}^t$ amount of electricity respectively, 
\begin{equation}
\label{eq:MP}
Pr[\mathcal{A}(Z^t) {\in} R] \leq \exp(\epsilon) Pr[\mathcal{A}(\hat{Z}^t) {\in} R],
\end{equation}
must hold.
Second, while in the original BP framework the dataset is constrained\cite{Blowfish}, in our formulation the constraints are on the probability of having a specific dataset.
This type of constraint is introduced as a probabilistic constraint in\cite{dphmm}, in contrast to a regular deterministic constraint. 
Keeping these slight differences in mind, we will develop a private mechanism in the next part.

%***************************************************************Mechanism********************************************************
\subsection{Rate Publication Mechanism} \label{mechanism}
The proposed Blowfish private mechanism is provided in Algorithm~\ref{base}.
In this mechanism, the optimal rate is perturbed before being published, by adding noise sampled from the Laplace distribution, as suggested by\cite{DPbook}. 
The main novelty is how to calculate a proper noise variance to ensure privacy at each time interval while not compromising the utility.
In particular, the noise variances are chosen to account for the plausible correlation considered in the model class.

% ---------- details ----------
This algorithm iterates over the following steps at each time:
For each possible model $\theta{\in}\Theta$, the algorithm calculates the priors based on Equation~\eqref{eq:pri} (Lines~\ref{alg:pri1}-\ref{alg:pri2}). 
Accordingly, it forms the set $\mu^{t,\theta}$ of all states satisfying constraints in~\ref{Blowfish} (Line~\ref{alg:mu}). 
Searching over all Blowfish neighbors (defined as $\nu^{t,\theta}$ in Line~\ref{alg:nu}), a set of houses each discriminating at least one pair of Blowfish neighbors is formed (Line~\ref{alg:k}). 
This set is used to calculate the required noise deviation to guarantee privacy for $\theta$ (Line~\ref{alg:noise}). 
Repeating over all $\theta\in\Theta$, a privacy-preserving noise variance under every possible model is identified (Line~\ref{alg:propnoise}).
The optimal rate is perturbed by adding noise sampled from the specified distribution and the result is published (Lines~\ref{alg:pert1}-\ref{alg:pert2}).
Finally, the posterior is calculated to be used in the next time step (Line~\ref{alg:post}).
% ----- Algorithm
\begin{algorithm}[H]
\caption{Rate Publication Mechanism}\label{base} 
\textbf{Input:} privacy parameter $\epsilon$, time horizon $T$, rate function parameters $(\alpha,\beta)$, consumption bound for each household $u_i$, set of data generating models $\Theta$ \nonumber \noindent
\\
\textbf{Output:} Published rates $\hat{r}^1,\cdots,\hat{r}^T$\nonumber \noindent
\begin{algorithmic}[1] 
\State At each time step $t\in\{1,\cdots,T\}$:

\For{$\theta=(\mathbf{A}^{t, \theta},\Pi^ \theta) \in \Theta$} \label{alg:select}

\State \label{alg:pri1} initialize: $P^{1|0, \theta} \gets \Pi ^ \theta$ \Comment{Prior at step 1}

\State \label{alg:pri2} $P^{t|t{-}1, \theta} \gets P^{t{-}1|t{-}1, \theta}\mathbf{A}^{t, \theta}$ \Comment{Calculate prior} 

\State \label{alg:mu} $\mu^{t,\theta} \gets \{ M^i|P_i^{{t|t{-}1}, \theta} {\neq} 0\}$ \Comment{Non-zero prior states}

\State \label{alg:nu} $\nu^{t,\theta} \gets \{(M^i,M^j) \in \mu^{t,\theta} {\times} \mu^{t,\theta} \; | \; ||M^i{-}M^j||_1 = 1 \}$

\State \label{alg:k} $\kappa^{t,\theta}  \gets \{\textit{Entries that differ between pairs in }\nu^{t,\theta}\}$

\State \label{alg:noise} $\lambda^{t,\theta} \gets \alpha {\times} \max_{h \in \kappa^{t,\theta}} u_h$
\EndFor
\State \label{alg:propnoise} $\lambda^{t} \gets \max_{\theta \in \Theta} \lambda^{t,\theta}$	\Comment{Maximum over all models}
\State \label{alg:pert1} $r^t \gets \alpha \sum_{i=1}^N X^t_i + \beta$					\Comment{Optimal rate}
\State $N^t \sim Lap(\lambda^t / \epsilon)$						 		\Comment{Noise to be added}
\State $\hat{r}^t \gets r^t+N^t$ 										\Comment{Noisy rate}\\
\Return \label{alg:pert2} $\hat{r}^t$														\Comment{Publish rate}
\State \label{alg:post} Derive $ P^{t|t,\theta}~\forall \theta \in \Theta$ by Equation~\eqref{eq:post}
\State Go to time step \textit{t+1}
\end{algorithmic}	
\end{algorithm}

% Theorem
Our first result below shows privacy is guaranteed in each time step and the second one considers privacy promise after releasing the rates $T$ times.
\begin{theorem}
	\label{main}
	The proposed rate publication mechanism preserves ($\epsilon,\mathcal{S}^t_{pairs},Q^t$)-BP at each time step. 
\end{theorem}
\begin{proof}
	\label{proof}
	We aim to calculate parameter  $\lambda^{t,\theta} \in \R_+$ such that adding noise sampled from $Lap(\lambda^{t,\theta} / \epsilon)$ to $r^t$ satisfies \eqref{eq:MP} for every pair of Blowfish neighbors under every possible CDHMM. 
	For each model characterized by $\theta\in \Theta$, we try to find the adequate noise variance that keeps the occupancy state of household $1$ private and then extend it to other users. 
	
	% neighbors
	Consider a Blowfish neighboring pair $(M^l,M^k)$. 
	Without loss of generality we can assume $M^l_1=1$, $M^k_1=0$, $M^l_{2:N}=M^k_{2:N}$.
	% other values
	The released rate is a linear combination of $X^t_i$ and $N^t$, $\hat{r}^t=\alpha \sum_{i=1}^N X^t_i + \beta + N^t$. 
	Recalling the probability density of the Laplace distribution and the conditional probability density of $X^t_i$, $f(X^t_i|S^t_i,\hat{r}^{t'})$, we derive the probability of observing a rate in a range $R$ conditioned on the state.
	For every instantiation $x_i^t$ of the random variables $X_i^t, ~\forall i \in \{1,\cdots,N\}$,
	\begin{IEEEeqnarray} {C}
	\label{eq:pr1}
	Pr[\hat{r}^t \in R |S^t=M^l] = 
	\int \limits_{0}^{u_N}{\cdots} \int \limits_{0}^{u_1}   \prod_{\ell=1}^{N}f(x^t_\ell|M^l_\ell,\hat{r}^{t'})\nonumber \\
	\int \limits_{R} \frac{\epsilon}{2\lambda^{t,\theta}_1} exp(-\epsilon|\hat{r}^t-\alpha \sum_{i=1}^N x^t_i{-}\beta| / \lambda^{t,\theta}_1) d_{\hat{r}^t}
	d_{x^t_1} ... d_{x^t_N}.%
	\end{IEEEeqnarray}
	
	Let $g(\hat{r}^t)\mathrel{\mathop:}=(\hat{r}^t{-}\alpha (\sum_{i=2}^N x^t_i{+}\frac{u_1}{2}){-}\beta)/\alpha$. It follows that
	\begin{IEEEeqnarray}{C}
	|g(\hat{r}^t)|{+}|x^t_1{-}\frac{u_1}{2}| \leq
	|g(\hat{r}^t){-} (x^t_1{-}\frac{u_1}{2})| \leq 
	|g(\hat{r}^t)|{-} |x^t_1{-}\frac{u_1}{2}| \nonumber \\
	\Rightarrow 
	|g(\hat{r}^t)|{+}\frac{u_1}{2} \leq 
	|g(\hat{r}^t){-} (x^t_1{-}\frac{u_1}{2})| \leq
	|g(\hat{r}^t)|{-}\frac{u_1}{2}. \label{eq:pr3}
	\end{IEEEeqnarray}
	
	In the above, the first inequality is due to applying the triangle inequality and the second one relies on $x^t_1\in[0,u_1]$.
	Substituting \eqref{eq:pr3} in \eqref{eq:pr1} obtains	\footnote{The terms $e^x$ and $exp(x)$ are used interchangeably.}
	\begin{IEEEeqnarray}{C}
	\int \limits_{0}^{u_N}{\cdots}\int \limits_{0}^{u_1} 
	\prod_{\ell=1}^{N}f(x^t_\ell|M^l_\ell,\hat{r}^{t'})
	\int \limits_{R}
	e^{\frac{-\epsilon\alpha}{\lambda^{t,\theta}_1}(g(\hat{r}^t){+}\frac{u_1}{2})} d_{\hat{r}^t} d_{x^t_1} ... d_{x^t_N}  \nonumber
	\\
	\leq 2\lambda^{t,\theta}_1 Pr[\hat{r}^t \in R|S^t=M^l] / \epsilon 	\leq \nonumber 
	\\ 
	\int \limits_{0}^{u_N}{\cdots}\int \limits_{0}^{u_1} 
	\prod_{\ell=1}^{N}f(x^t_\ell|M^l_\ell,\hat{r}^{t'})
	\int \limits_{R}
	e^{\frac{-\epsilon\alpha}{\lambda^{t,\theta}_1}(g(\hat{r}^t){-}\frac{u_1}{2})} d_{\hat{r}^t} d_{x^t_1} ... d_{x^t_N}. \nonumber	
	\end{IEEEeqnarray}
	
	All the above integrals, except the one on $x^t_1$, and $g(\hat{r}^t)$ are independent of $x^t_1$. 
	We remove the integral on $x^t_1$ by integrating $f(x^t_1|M^l_1,\hat{r}^{t'})$ over $0$ to $\infty$, yielding $1$.
	\begin{IEEEeqnarray}{C} % x1 removed
	\int \limits_{0}^{u_N}{\cdots}\int \limits_{0}^{u_1} 
	\prod_{\ell=2}^{N}f(x^t_\ell|M^l_\ell,\hat{r}^{t'})
	\int \limits_{R}
	e^{\frac{-\epsilon\alpha}{\lambda^{t,\theta}_1}(g(\hat{r}^t){+}\frac{u_1}{2})} d_{\hat{r}^t} d_{x^t_2} ... d_{x^t_N}  \nonumber \\
	\leq 2\lambda^{t,\theta}_1 Pr[\hat{r}^t \in R|S^t=M^l]/\epsilon 	\leq \nonumber \\ 
	\int \limits_{0}^{u_N}{\cdots}\int \limits_{0}^{u_1} 
	\prod_{\ell=2}^{N}f(x^t_\ell|M^l_\ell,\hat{r}^{t'})
	\int \limits_{R}
	e^{\frac{-\epsilon\alpha}{\lambda^{t,\theta}_1}(g(\hat{r}^t){-}\frac{u_1}{2})} d_{\hat{r}^t} d_{x^t_2} ... d_{x^t_N}. \nonumber
	\end{IEEEeqnarray}
	
	As $M^l_i=M^k_i ~\forall i\in\{2,\cdots,N\}$, repeating the procedure for $\hat{S}^t$ and dividing the emerged inequalities, we get
	\begin{IEEEeqnarray} {C}
	\label{eq:pr5}
	e^{\frac{-\epsilon\alpha}{\lambda^{t,\theta}_1} u_1}
	\leq \frac{Pr[\hat{r}^t \in R |S^t=M^l]}{Pr[\hat{r}^t \in R |\hat{S}^t=M^k]}
	\leq e^{\frac{\epsilon\alpha}{\lambda^{t,\theta}_1} u_1}.
	\end{IEEEeqnarray} 
	To keep data of user \textit{1} private, \eqref{eq:pr5} must be bounded by $e^{- \epsilon}$ and $e^\epsilon$, that is equivalent to $ \alpha u_1 \leq \lambda^{t,\theta}_1$.
	
	As shown for $(M^l,M^k)$, if datasets of a neighboring pair differ in the occupancy of  house $i$, adding noise sampled from $Lap(\lambda^{t,\theta}_i / \epsilon)$ with $\lambda^{t,\theta}_1 {\geq} \alpha u_i$ suffices. 
	Accordingly, a noise distribution of $Lap(\lambda^{t,\theta}/ \epsilon)$ with $\lambda^{t,\theta} =\alpha \max_{h \in \kappa^{t,\theta}} u_h$ ensures privacy given model $\theta$. 
	We repeat this procedure $\forall \theta\in\Theta$ and use the largest noise deviation to provide privacy if the data was generated by each model in the given class. 
\end{proof}

%compos
As we have developed a real-time mechanism, potential information acquisition due to future releases does not directly affect the noise distribution at the current time interval.
We propose the following theorem to study the mechanism's privacy observing the released rates over $T$ time steps.
\begin{theorem}\label{compos}
	Algorithm~\ref{base} is ($\epsilon T,\mathcal{S}_{pairs},Q$)-Blowfish private observing $T$ rates where $ \mathcal{S}_{pairs} = (\mathcal{S}_{pairs}^1,...,\mathcal{S}_{pairs}^T)$ and $Q = (Q^1,...,Q^T)$.
\end{theorem}  
\begin{proof}
	Assume rates $\hat{r}^1,\cdots,\hat{r}^T$ are published by Algorithm \ref{base}.
	Let $S=(S^1,\cdots,S^T)$ and $\hat{S}=(\hat{S}^1,\cdots,\hat{S}^T)$ be Blowfish neighbors ${\forall}t\in\{1,\cdots,T\}$.
	The probability of being at states $S$ and observing rates $\hat{r}^t \in R^t$ is
	\begin{IEEEeqnarray}{C}
	\label{eq:prob}
	Pr[S,\hat{r}^1 \in R^1,\cdots,\hat{r}^T \in R^T]\nonumber \\
	=\pi_{S^1}\prod_{t=2}^T Pr[S^t|S^{t{-}1}] 
	\prod_{t=1}^T Pr[\hat{r}^t \in R^t|S^t] =Pr[S] \nonumber \\
	\times \prod_{t=1}^T Pr[\hat{r}^t \in R^t|S^t] \leq Pr[S]\prod_{t=1}^T e^\epsilon Pr[\hat{r}^t \in R^t|\hat{S}^t] \label{eq:comp},
	\end{IEEEeqnarray}
	where~\eqref{eq:comp} is true because the mechanism is ($\epsilon,\mathcal{S}^t_{pairs},Q^t$)-Blowfish private at each step. Thus, $Pr[\hat{r}^1 \in R^1,\cdots,\hat{r}^T \in R^T|S]/Pr[\hat{r}^1 \in R^1,\cdots,\hat{r}^T \in R^T|\hat{S}]\leq e^{\epsilon T}$.
\end{proof}
A similar approach to quantify the information leakage due to answering a query multiple times is sequential composition theorem\cite{DPbook}\cite{Blowfish}. 
Even though Theorem \ref{compos} yields exactly the same results as applying the composition theorem, it is important to note that the latter one cannot be used because unlike the setting in [1], in RTP the queries are not answered on the same dataset (both $S^t$ and $Z^t$ may change by time).  
%*******************************************************************************
%********************************Simulation*************************************
%*******************************************************************************
\section{SIMULATION} \label{sim}
We evaluate our proposed algorithm on a synthetic dataset obtained by simulating a non-homogeneous CDHMM, formulating the relationship between the occupancy state, the power consumption, and the electricity rates. 
For simplicity, $\Theta$ is assumed to be a single model in this case.
We assume that all households' consumption data at a given time step $t$ are independent, thereby reducing the problem to that of simulating $N=1000$ individual CDHMMs.

\begin{figure*}[tb]
	\centering
	\begin{subfigure}{0.3\linewidth}
		\includegraphics[width=\linewidth]{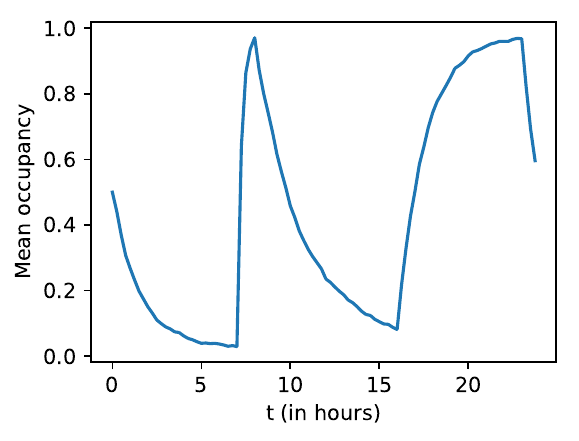}
		\caption{Mean occupancy}
		\label{fig:sim_mc}
	\end{subfigure}
	\begin{subfigure}{0.3\linewidth}
		\includegraphics[width=\linewidth]{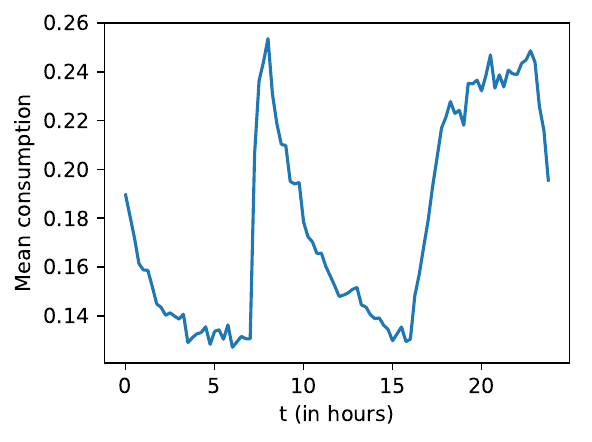}
		\caption{Mean consumption}
		\label{fig:sim_obs}
	\end{subfigure}
	\begin{subfigure}{0.3\linewidth}
		\includegraphics[width=\linewidth]{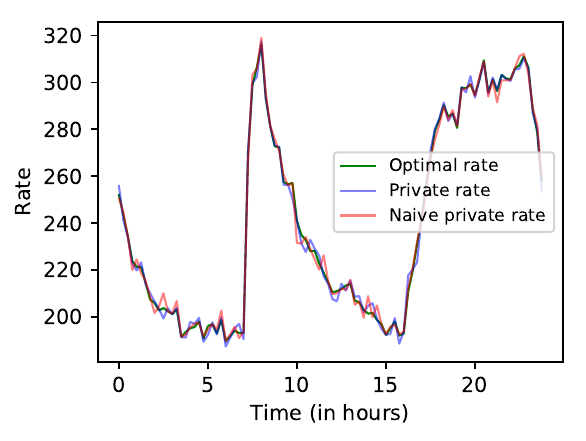}
		\caption{Calculated rates}
		\label{fig:sim_rate}
	\end{subfigure}
	\caption{Simulation outcomes}
\end{figure*}

\subsection{Simulation procedure}
Each Markov chain is given by a time-dependent transition matrix.
The probabilities in this matrix are obtained by perturbing a base matrix $\mathbf{A}^{t,\theta}$ with additive Gaussian noise in order to simulate differences between habits of individuals. 
For simplicity, we divide the day into four intervals (morning 7-8 AM, noon 8 AM-16 PM, evening 16-23 PM, and night 23 PM-7 AM) during which the transition probabilities remain constant. Hence, there are only four unique matrices $\mathbf{A}^{t,\theta_1}, \cdots, A^{t,\theta_4}$.
The transition probabilities are chosen such that there is
a 5$\%$ chance to leave the house and a 98$\%$ chance to wake up at some point during the morning,
a 95$\%$ chance to leave the house and a 5$\%$ chance to come home at some point during noon,
a 5$\%$ chance to leave the house and a 99$\%$ chance to come home at some point during the evening,
a 99.5$\%$ chance to go to sleep and a 1$\%$ chance to get up at some point during the night.
Note that these numbers are mere guesses.
We simulate this Markov chain for 24 hours, where each hour is divided into four time steps (i.e. $T=96$). 
Fig.~\ref{fig:sim_mc} shows the mean occupancy of all households in the simulation. The four periods of the day can be recognized clearly by the sharp occupancy changes.

For each CDHMM, the observations are modeled by uniform distributions $X^t_i \sim U(0, u_i)$ where $u_i \sim U(0, u_{\max})$. We choose $u_{\max} {=} 1$ when the household is occupied and $u_{\max} {=} 0.5$ otherwise.
The observation probability distribution can be calculated by~\eqref{eq:pr1} accordingly.
Fig.~\ref{fig:sim_obs} shows the mean consumption of all households in the simulation.

\subsection{Results}
Based on the simulated consumption, we compute the optimal rates according to Equation~\eqref{eq:fr}, as well as the Blowfish-private rate by Algorithm~\ref{base} with $\epsilon = 0.5$ in Definition~\ref{privacy}.
% alpha
The generation cost, $J(Z^t)$, is a quadratic function in~\cite{cost} with $\alpha {=} 1, \beta {=} 62.5$ (for simplicity, the equation is scaled).

% naive
We also contrast the utility of the Blowfish mechanism with that of a naive Laplace mechanism, which does not accommodate any correlation model and consequently sets no constraints.
Eliminating the constraints results in $\kappa^{t, \theta}$, in Algorithm~\ref{base} to be the set of all houses and results in a noise distribution of $Lap(\alpha u_{\max} / \epsilon)$ in every time step.
Hence, disregarding the model reduces the privacy notion introduced in~\ref{Blowfish} to group DP~\cite{DPbook}, which assumes the entire time steps are fully correlated. 
We refer to the outputs from the aforementioned mechanism by \textit{naive private rates} and the rates obtained from Algorithm~\ref{base} by \textit{private rates} in Fig.~\ref{fig:sim_rate}.

The discriminating term between the two mechanisms is the set over which the maximum consumption bound of users is calculated.
While in the naive mechanism the search for the maximum bound is done over all houses, in Algorithm~\ref{base} only a subset of houses with non-zero prior are examined.
From Fig.~\ref{fig:sim_rate}, it can be seen that the Blowfish mechanism uses a smaller noise parameter than the naive mechanism in some time steps.

The deviation from the optimal rate in Fig.~\ref{fig:sim_rate} is measured in terms of the Root Mean Squared Relative Errors (RMSRE), 
$RMSRE \mathrel{\mathop:}= \frac{1}{T}\sqrt{\sum_{t=1}^T ((\hat{r}^t - r^t)/r^t)^2}.$
For our chosen parameter values we obtain a RMSRE of $1.089e{-}3$ for the Blowfish mechanism and $1.149e{-}3$ for the naive mechanism, which is $5.5\%$ more. 
Both errors are considered low compared to the observed rate range. 
%Their relative significance would vanish as $N$ grows since the rate increases linearly with $N$ while the noise variance remains unchanged. 
% Though, it should be reiterated that the privacy guarantees are with respect to a single time step. 

%%%%%%%%%%%%%%%%%%%%%%%%%%%%%%%%%%%%%%%%%%%%%%%%%%%%%%%%%%%%%%%%%%%%%%%%%%%%%%%%
\section{Conclusions}
We proposed a privacy-preserving mechanism to publish electricity rates over correlated consecutive time steps dynamically.
Our approach is based on the notion of Blowfish privacy.
We derived constraints so that if there is enough evidence showing a house is not occupied, this house is removed from the set of protected houses.
This in turn improves the choice of perturbation variance for maximized utility. 
Furthermore, we showed that the noise variance used by the proposed mechanism is at least as small as the one obtained with the only alternative in~\cite{DPbook} and can be smaller in general. 
We validated the performance improvement with respect to this alternative solution on a synthesized dataset.
Our current work aims to explore potential mechanisms that ensure the privacy degradation is sub-linear with the time horizon. 
\addtolength{\textheight}{-4cm}

%***********************************References***********************************
\bibliographystyle{IEEEtran}
\scriptsize{
	\bibliography{IEEEabrv,references}
}

\end{document}